\documentclass[journal]{IEEEtran}
\usepackage{amsmath,amsfonts}
\usepackage{algorithmic}
\usepackage{algorithm}
\usepackage{array}
\usepackage[caption=false,font=normalsize,labelfont=sf,textfont=sf]{subfig}
\usepackage{textcomp}
\usepackage{stfloats}
\usepackage{url}
\usepackage{verbatim}
\usepackage{graphicx}
\usepackage{cite}
\usepackage{braket}
\usepackage{cuted}
\usepackage{mathtools}
\usepackage{amsthm}
\usepackage{multirow}
\usepackage{hyperref}
\usepackage{bbm}
\usepackage{relsize}
\usepackage{blkarray}

\usepackage{accents} 
\usepackage{xcolor}
\usepackage{nicematrix}
\usepackage{quantikz}

\usepackage{bm}     % edi added

\def\ET{\smallbreak\par\noindent}

\let\q=\quad
\def\vq{\;,\q}

\let\hb=\hbox

\let\Arr=\longrightarrow

\def\B(#1){\mbox{${\bf #1}$}}
\def\BB(#1){\mbox{${\bm #1}$}}
\def\C(#1){{\cal #1}}

\let\q=\quad

\def\vl{\;,\;}

\def\removelastskip{\ifdim\lastskip=0pt \else\vskip-\lastskip\fi}

\def\evv#1{\subsection*{#1}}

\def\ket#1{|#1\rangle}

\def\TT{{\hbt{T}}}
\def\with{\quad\hb{with}\quad}

\def\arrDFT{\;\stackrel{\rm DFT}{\Arr}\;}

\def\arrdvQFT{\;\stackrel{\rm dvQFT}{\Arr}\;}
\def\arrcvQFT{\;\stackrel{\rm cvQFT}{\Arr}\;}

\def\bma{\begin{bmatrix}} 
	\def\ema{\end{bmatrix}} 

\newcommand{\ba}{\begin{array}}
	\newcommand{\ea}{\end{array}}
\newcommand{\beq}{\begin{equation}}
	\newcommand{\eeq}{\end{equation}}
\newcommand{\beqa}{\begin{eqnarray}}
	\newcommand{\eeqa}{\end{eqnarray}}

\def\hbs#1{\hbox{\rm\small #1}}
\def\hbf#1{\hbox{\rm\footnotesize #1}}
\def\hbt#1{\hbox{\rm\tiny #1}}
\def\met{\frac12}
\def\mett{\frac12}
\let\D=d

\def\M(#1){\mathbb{#1}}

\def\ov{\overline}

	\makeatletter
	\def\cases#1{\left\{\,\vcenter{\normalbaselines\m@th
			\ialign{$##\hfil$&\quad##\hfil\crcr#1\crcr}}\right.}
	\makeatother
	\def\E{\hb{e}}
	\def\I{\hbs{i}}
	
	\def\texIT#1{\indent\llap{#1\enspace}\ignorespaces}
	\def\IT#1 {\smallbreak\hangindent\parindent\texIT{#1}
		\ifdim\lastskip=\smallskipamount\removelastskip\penalty55\
		\smallskip\fi}
	\catcode`@=11
	\newdimen\jot \jot=3pt
	\def\openup{\afterassignment\@penup\dimen@=}
	\def\@penup{\advance\lineskip\dimen@
		\advance\baselineskip\dimen@
		\advance\lineskiplimit\dimen@}
	\def\eqalign#1{\null\,\vcenter{\openup\jot\m@th
			\ialign{\strut\hfil$\displaystyle{##}$&$\displaystyle{{}##}$\hfil
				\crcr#1\crcr}}\,}
	\newif\ifdt@p
	\def\eqal(#1)#2{\ifbozze\xdef\etichetta{#1}\fi\eqlabel{#1}\eqalignno{#2}}
	\makeatletter
	\@ifundefined{qmatrix}
	{\newcommand{\qmatrix}[1]{\begin{bmatrix} #1 \end{bmatrix}}}
	{\renewcommand{\qmatrix}[1]{\begin{bmatrix} #1 \end{bmatrix}}}
	\makeatother

	\usepackage{amsthm}
	\theoremstyle{definition}
	\makeatletter
	\let\olddefinition\definition % salva l'eventuale macro
	\let\definition\relax         % libera il nome
	\newtheorem{definition}{Definition}
	\let\definition\olddefinition % (opzionale) ripristina la macro col vecchio nome
	\makeatother

\newtheorem{ndefinition}{Definition}
\newtheorem{nproposition}{Proposition}
\newtheorem{ntheorem}{Theorem}

\begin{document}
	
	\title{The Quantum Fourier Transform for Continuous Variables}
	
	\author{Gianfranco~Cariolaro,
		Edi~Ruffa, Amir~Mohammad~Yaghoobianzadeh, and
		Jawad~A.~Salehi, \itshape Fellow, IEEE \normalfont %
		\thanks{G.~Cariolaro is with the Department of Information Engineering, University of Padova, Via Gradenigo 6/B, 35131 Padova, Italy (e-mail: cariolar@dei.unipd.it).}%
		\thanks{E.~Ruffa is with Vimar SpA, Via IV Novembre 32, 36063 Vicenza, Italy (e-mail: edi.ruffa@ieee.org).}%
		\thanks{J.~A.~Salehi and A.~M.~Yaghoobianzadeh are with the Sharif Quantum Center, Electrical Engineering Department, Sharif University of Technology, Tehran, Iran (e-mails: jasalehi@sharif.edu; am.yaghoobianzadeh@gmail.com).}%
	}
	
	%% The paper headers
	%\markboth{}
	%
	%\IEEEpubid{}
	%% Remember, if you use this you must call \IEEEpubidadjcol in the second
	%% column for its text to clear the IEEEpubid mark.
	
	\maketitle
	
	\begin{abstract}
		The quantum Fourier transform  for discrete
		variable (dvQFT) is an efficient algorithm  for
		several applications. It is
		usually considered for the processing of quantum bits (qubits) and
		its efficient implementation is obtained with two elementary components:
		the Hadamard gate and the controlled--phase gate.
		In this paper, the quantum Fourier transform operating with continuous variables (cvQFT) is considered. Thus, the
		environment becomes  the Hilbert space, where the natural definition of
		the cvQFT will be related to rotation operators, which
		in the $N$--mode are completely specified by unitary matrices of order $N$. Then the cvQFT is defined as the
		rotation operator whose rotation matrix is given by the discrete Fourier transform (DFT) matrix.
		For the implementation of rotation operators with primitive components
		(single--mode rotations and beam splitters), we follow the well known Murnaghan procedure, with appropriate modifications.  Moreover, algorithms related
		to the fast Fourier transform  (FFT) are applied to reduce drastically
		the implementation complexity.
		The final part is concerned with the application of the cvQFT to general Gaussian states. In particular, we show that cvQFT has the simple effect of transforming the displacement vector by a one-dimensional DFT, the squeeze matrix by a two-dimensional DFT, and the rotation matrix by a Fourier-like similarity transform.
	\end{abstract}
	

	\section{\label{IN} INTRODUCTION}
	
	The quantum Fourier transform (QFT) for discrete
	variable (dvQFT) is an efficient algorithm  for
	several applications, as factoring, simulations of quantum systems,
	quantum chaos,  quantum tomography, and several other
	applications  \cite{Niel2000,Weistein2001}.
	The dvQFT is usually considered for the processing of quantum bits (qubits) and
	its efficient implementation is obtained with two elementary components:
	the Hadamard gate and the controlled--phase gate.
	
	In this paper  we consider the QFT operating with continuous variables
	(cvQFT) and in particular with Gaussian states. Thus, the
	environment becomes  the Hilbert space, where the natural definition of
	the cvQFT will be related to  rotation operators, which
	in the $N$--mode are completely specified by unitary matrices of order $N$. Then the quantum Fourier transform
	for continuous variables (cvQFT) is defined as the rotation operator whose rotation matrix is given by the DFT matrix.

	The implementation of the rotation operators is strictly
	related to the factorization of the unitary complex matrices,
	since the set of $N$--mode rotation operators is isomorphic to
	the Lie matrix group of the unitary $N\times N$ complex
	matrices \cite{Hall2003}. A vast literature on the factorization of unitary
	matrices is available, but all researches on the topic have a
	purely mathematical interest, with the exception (at least in
	the authors' knowledge) of an often-cited letter 
	by Reck et al., \cite{Zeil94},
	which first tackled the practical problem of realizing linear 
	optical operators by simple components. An ideal goal would
	be to factorize the unitary matrix, and thence the rotation
	operator, into blocks depending on a single real number,
	corresponding to a simple linear operator, as a single--mode
	phase shifter or a two--mode real beam splitter. To this purpose, 
	as we shall see below, the 60--years-old mathematical
	approach by Murnaghan \cite{Murn52,Murn58} remains the most suitable
	method. However, Murnaghan's approach does not arrive at
	closed-form formulas, and so we have devised an appropriate
	algebra to get explicit results \cite{Cari18}.

	The paper is organized as follows. In Section II, we review the 
	discrete Fourier transform (DFT), and also the dvQFT,
	just for comparison. In Section III, we introduce the cvQFT and recall rotation operators and related unitary phase matrices; The rest of the paper consists of two parts. Part I is concerned with the implementation of the cvQFT.
	In Section IV, we recall the modified Murnaghan procedure for the
	implementation of rotation operators with primitive components
	(single--mode rotations and beam splitters). Also, we apply
	the modified Murnaghan procedure  for the implementation of the 4--cvQFT; this 
	case is sufficient to provide a glimpse on the high complexity of the general case. However, there are
	several procedures to reduce the complexity, 
	as the use of Kronecker product \cite{Camps2020} or expressing the indexes
	in  binary form. In Section V, we apply an original method
	of complexity reduction based on the techniques of Digital Signal Processing (DSP) of the Unified Signal Theory \cite{CariUST}. Part II is about Gaussian states. Section VI is concerned with the application of the cvQFT to Gaussian states.
	In Section VII we evaluate how the covariance matrix is modified after the application of the cvQFT.

	\section{\label{DFT} THE DFT AND THE {dv}QFT}
	The discrete Fourier transform (DFT) acts on a vector of complex numbers
	$\B(s)=[s_0,s_1,\ldots s_{N-1}]^\TT$ and produces a complex vector
	$\B(S)=[S_0,S_1,\ldots S_{N-1}]^\TT$ defined by
	\beq
	S_k= \frac1{\sqrt N}\sum_{n=0}^{N-1}\;  s_n\; e^{i 2\pi{kn}/N}
	\label{U2a}
	\eeq
	The inverse DFT (IDFT) recovers the vector $\B(s)$ from the vector  
	$\B(S)$ according to
	\beq
	s_n= \frac1{\sqrt N}\sum_{k=0}^{N-1}\;  S_k\; e^{-i 2\pi{kn}/N}
	\label{U2b}
	\eeq
	With the introduction of the DFT matrix
	\beq
	\B(W)_N= \left[w_{rs}\right]_{r,s=0,1,\ldots,N-1} \with 
	w_{rs}=\frac1{\sqrt N} e^{i  2\pi rs/N}
	\label{B1}
	\eeq
	Eq.~(\ref{U2a}) becomes
	\beq
	\B(S)= \B(W)_N\;\B(s)\vq \B(s)=\B(W)_N^{-1}\;\B(S)
	\eeq
	Note that the DFT matrix is unitary:  $\B(W)_N^{-1}=\B(W)_N^*$.
	
	The brute--force application of the DFT of order $N$ has a computational
	complexity of $N^2$ operations. When $N$ is a power of 2, the fast algorithm fast Fourier transform, FFT, reduces the complexity to $N\log_2 N$ operations. 
	
	The dvQFT on an orthonormal basis $\ket 0, \ket 1,\ldots,\ket{N-1}$ is
	a linear unitary operator with the following action on the basis states
	\cite{Brown2006}, \cite{Mastriani2021}
	\beq
	\ket k \arrdvQFT\frac1{\sqrt N}\sum_{n=0}^{N-1}\;   e^{i 2\pi{kn}/N}\ket n;
	\eeq
	Equivalently, the action on an arbitrary state can  be written as
	\beq
	\ket s=\sum_{n=0}^{N-1} s_n \ket n\arrdvQFT \ket S=\sum_{k=0}^{N-1} S_k\ket k
	\eeq
	where 
	\beq
	\B(s)=[s_0,\ldots  s_{N-1}]  \arrDFT  \B(S)=[S_0,\ldots,S_{N-1}] 
	\eeq
	In words: in the dvQFT the coefficients (probability amplitudes) of the output state $\ket S$
	are given by DFT of the coefficients of the input states.
	Note that the input state  $\ket s$ is usually given by a sequence of 
	$N$ qubits,  $\ket s=\ket {s_0}\cdots\ket{s_{N-1}}$
	with $\ket{s_i}\in \text{span}\{\ket 0,\ket 1\}$.

	\paragraph{Implementation of dvQFT} For the implementation of the dvQFT  two gates are used:
	the Hadamard gate and the controlled--phase gate. The graphical symbols for these gates are given in Fig.~\ref{FF52}.
	\begin{figure}[!t]
		\includegraphics[scale=0.9]{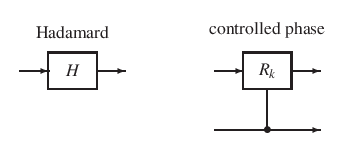}
		\centering
		\caption{The Hadamard and the controlled--phase gates.}
		\label{FF52}
	\vspace*{0pt}
	\end{figure}	
	The Hadamard gate
	acts on a single qubit.  It is represented by the Hadamard matrix
	\beq
	\B(H)=\frac{1}{\sqrt2}
	\qmatrix{1 & 1\cr
		1&-1\cr}=\B(W)_2
	\label{H6B}
	\eeq
	that is by the 2--DFT matrix. It maps the input qubit as follows
	\beq
	\ket0\mapsto \frac{1}{\sqrt2}(\ket0+\ket1)\vq
	\ket1\mapsto \frac{1}{\sqrt2}(\ket0-\ket1)
	\eeq

	The block $\B(R)_k = \B(R)(\frac{2\pi}{2^k})$ is a  controlled--phase gate, where it is described by the matrix
	\beq
	\B(R)(\phi)=
	\qmatrix{1&0&0&0 \cr
		0&1&0&0\cr
		0&0&1&0\cr
		0&0&0&e^{i\phi}\cr}
	\eeq
	With respect to the reference basis it 	shifts  by  $\phi$ only when the input
	is $\ket{1}\ket{1}$
	\beq
	\ket{a,b}\;\mapsto\; \cases{ e^{i\phi}\ket{a,b}& \text{for $a=b=1$}\cr
		\ket{a,b}       & \text{otherwise}\cr}
	\eeq
	\ET The global scheme is illustrated in Fig.~\ref{FF50}.
	
	\begin{figure*}[!t]
		\includegraphics[scale=1.1]{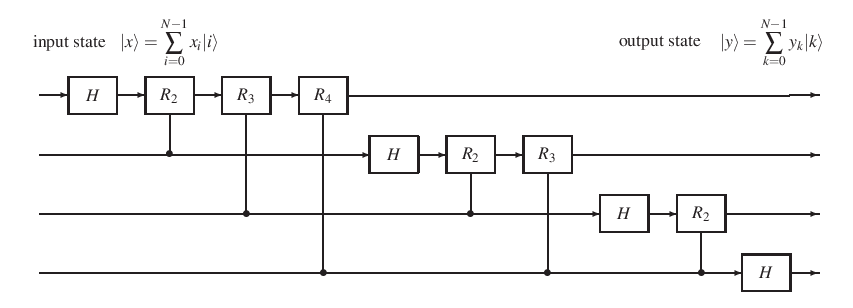}
		\centering
		\caption{Implementation of the dvQFT for $N=4$ according to 
		Ref. \cite{Niel2000}.
		The coefficients $x_i$ and $y_k$ are related by the 4--DFT.}
		\label{FF50}
	\hrulefill
	\vspace*{0pt}
	\end{figure*}

	\section{\label{DF} DEFINITION OF QFT FOR CONTINUOUS VARIABLES (cvQFT)}

	The definition is formulated in terms of quantum rotation operators. Then
	we recall that a rotation operator in the $N$--mode 
	Hilbert space $\C(H)^N$ has the form
	\beq
	R(\BB(\phi))=\E^{\I \,\B(a)^*\BB(\phi)\B(a)}  \label{A2}
	\eeq
	where  $\BB(\phi)$ is an $N\times N$ Hermitian matrix and $\B(a)$ collects
	the $N$ annihilation operators.
	The corresponding Bogoliubov transformation is given by
	\beq
	R^*(\BB(\phi))\,\B(a)\,R(\BB(\phi))=\E^{\I\BB(\phi)}\B(a)
	\label{A4}
	\eeq
	The $N\times N$ unitary matrix associated to the rotation operator
	\beq
	\B(U)_{\BB(\phi)}=\E^{\I\BB(\phi)}
	\label{A6}
	\eeq
	completely specifies the rotation operator $R(\BB(\phi))$.
	Given the matrix $\B(U)_{\BB(\phi)}$, relation (\ref{A6}) uniquely identifies the
	phase  matrix $\BB(\phi)$, see \cite{MaRh90},
	but the evaluation of $\BB(\phi)$ is not necessary
	because every application will work only in terms of the matrix $\B(U)_{\BB(\phi)}$,
	as is in the  Bogoliubov transformation (\ref{A4}).
	
	We are now ready for the definition:
	
	\begin{ndefinition}
		The quantum Fourier transform for continuous
		variables (cvQFT) is the transformation in the Hilbert space $\C(H)^N$
		performed by a rotation
		operator whose unitary matrix $\B(U)_{\BB(\phi)}$ is the DFT matrix 
		\beq
		\B(U)_{\BB(\phi)}= \E^{\I\phi_{\hbt{DFT}}} = \B(W)_N
		\eeq
		The inverse transformation (IcvQFT) is performed by a rotation operator 
		whose rotation matrix is the IDFT matrix
		$\B(U)_{\BB(\phi)}^*=\B(W)_N^*=\B(W)_N^{-1}$.
		There have been exploited similar expressions in other contexts 
		\cite{rezai2021quantum, barak2007quantum, tabia2016recursive}.
		\label{D2}
	\end{ndefinition}

	The application of the cvQFT to a pure 
	quantum state $\ket{\gamma}\in\C(H)^N$
	provides the transformation 
	\beq
	\ket{\gamma} \arrcvQFT  \ket{\gamma_{\hbt{QFT}}}= R(\phi_{\hbt{DFT}})\,\ket{\gamma}\label{A8} 
	\eeq
	and for a mixed (noisy) state the transformation is
	\beq
	\rho\arrcvQFT  \rho_{\hbt{QFT}}= R(\phi_{\hbt{DFT}})\,\rho\,R^*(\phi_{\hbt{DFT}}) \label{A9}
	\eeq
	as illustrated in Fig.~\ref{FF51}.
	\begin{figure}[!t]
		\includegraphics[scale=0.9]{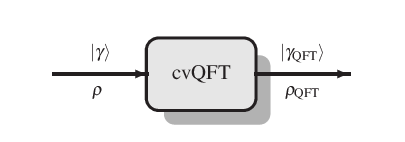}
		\centering
		\caption{Application of the cvQFT to a 
			pure quantum state and to 
			a mixed quantum state.}
		\label{FF51}
	\end{figure}

	In the simplest case the quantum state 
	may be an $N$--mode  displacement, but   it may be a Gaussian
	state  (squeezed+displacement), pure or mixed , and also 
	a non Gaussian state, e.g. a photon added Gaussian state  \cite{Cari2022}.
	
	As said above, the evaluation of the 
	rotation matrix $\phi_{\hbt{DFT}}$ such that $e^{i\phi_{\hbt{DFT}}}=\B(W)_N$ 
	has no relevance. However, 
	for curiosity,  the evaluation for the first orders gives
	\footnote{obtained with \texttt{ MatrixFunction[Log,$W_N$]} of \texttt{ Mathematica.}}:
	
	\begin{itemize}
		\item For $N=1$ 
		\beq
		\E^{\I\BB(\phi)} = \B(W)_1=[1]\q\to\q\BB(\phi)=[0]
		\eeq
		\item For $N=2$
		\begin{equation}
			\begin{split}
				\E^{\I\BB(\phi)} &= \B(W)_2= \frac{1}{\sqrt{2}}\left[
				\begin{array}{cc}
					1  & 1 \\
					1 & -1  
				\end{array}\right] \q\to\\
				\BB(\phi)&=\left[
				\begin{array}{cc}
					-\frac{1}{4} \left(-2+\sqrt{2}\right) \pi  & -\frac{\pi }{2 \sqrt{2}} \\
					-\frac{\pi }{2 \sqrt{2}} & \frac{1}{4} \left(2+\sqrt{2}\right) \pi  \\
				\end{array}\right]
			\end{split}
		\end{equation}
	\end{itemize}

	\section*{PART I: IMPLEMENTATION OF THE cvQFT }
	
	Considering the definition, the practical application of the cvQFT is essentially based on the implementation of the rotation operators with simple quantum components. As known, this problem is solved by a factorization of the associated unitary matrix, in such a way that each factor depends
	on a single real number. The corresponding theory, based on the
	Murnaghan  procedure, is recalled in the next
	section and leads to explicit results. However for $N$ large
	the Murnaghan procedure leads to very complicated structures.
	But in the cvQFT, the unitary matrix is given by the DFT  matrix. Then, with the help of digital signal processing (DSP), mainly the fast Fourier transform, we will find a very simple solution.

	\section{\label{CH} THE MURNAGHAN PROCEDURE}
	
	In this section we recall the Murnaghan approach of recursive factorization
	of a unitary matrix, which leads to the implementation of  rotation operators
	with elementary components. 
	We begin with the description of these components.

	\subsection{Primitive components  for the implementation}

	The primitive components, which  are illustrated in Fig. \ref{FF24}, are

	\begin{figure}[!h]
		\includegraphics[scale=0.9]{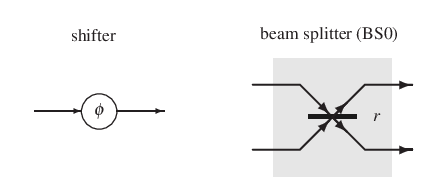}
		\centering
		\caption{Graphical representation of the two primitive 
		components.}
		\label{FF24}
	\end{figure}

	\IT 1) {\it phase shifters}, which are single--mode rotation operators,
	specified by a  scalar phase $\phi$,
	\IT 2) {\it free-phase beam splitters} (BS0), specified by the
	unitary matrix
	\beq   
	\B(U)_{\hbt{BS0}} = \qmatrix{t&r\cr r&-t\cr}
	\eeq
	where $r$ is the reflectivity and $t=\sqrt{1-r^2}$ is the transmissivity.

	\ET Two other elementary blocks are obtained from these primitive components;
	\IT i) {\it beam splitter with phases} (BS$\gamma$), 
	specified by a $2\times2$ unitary matrix, say
	\beq
	\B(U)_{\hbt{BS$\gamma$}}=\qmatrix{ te^{i\gamma_{11}} & re^{i\gamma_{12}} \cr re^{i\gamma_{21}} & -te^{i\gamma_{22}} \cr},
	\label{A21}
	\eeq
	where $\gamma_{21} = \gamma_{11} + \gamma_{22} - \gamma_{12}$.
	\IT ii) {\it $N$--input  $N$-output BSs with phase}
	($N$BS$\gamma$), which are essentially 
	BS with phase with $N-2$ extra connections. This unitary matrix is obtained by inserting in  the identity matrix of order $N$ the parameters of a BS$\gamma$.
	
	\ET To find the implementation of BS$\gamma$ with primitive components we recall \cite{Cari18}.

	\begin{nproposition}
		An arbitrary two-mode rotation operator,
		specified by the unitary matrix given by (\ref{A21}), 
		can be implemented by (1) two phase shifters with 
		phases $\gamma_{11}$ and  $\gamma_{12}$, followed by
		(2) a BS0
		with reflectivity $r$, followed by (3) a phase shifter  with phase $\mu=\gamma_{22}-\gamma_{12}$, as shown in Fig.~\ref{DE273}.
		\label{P2}
	\end{nproposition}

	\begin{figure}[!t]
		\includegraphics[scale=0.9]{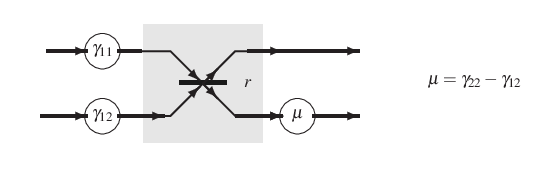}
		\centering
		\caption{The beam splitter with phase.}
		\label{DE273}
	\end{figure}

	\ET
	The proof is a consequence of  the orthogonality condition of 
	the matrix $\B(U)_{\hbt{BS$\gamma$}}$, which leads to the factorization 
	\beq 
	\B(U)_{\hbt{BS$\gamma$}}=\left[
	\begin{array}{cc}
		1 & 0 \\
		0 & e^{i (\gamma _{22}-\gamma _{12})} \\
	\end{array}\right]\left[
	\begin{array}{cc}
		t & r\\
		r & -t \\
	\end{array}\right]\left[
	\begin{array}{cc}
		e^{i \gamma _{11}} & 0 \\
		0 & e^{i \gamma _{12}} \\
	\end{array}\right]
	\label{A23}
	\eeq

	These blocks $N$BS$\gamma$, symbolized $\B(T)_{i,j}(r,\beta)_N$,
	depend on the order $N$ and 
	the indexes $i,j$ with $i,j=1,2,\ldots,N$ with  $j>i$, which denote the rows where 
	the BS$\gamma$ is inserted in the identity matrix.
	Their expressions are given,  for $N=3$
	\begin{subequations}
		\begin{align}
			\B(T)_{12}(r,\beta)_3=\left[
			\begin{array}{ccc}
				t &  re^{i \beta }  &0\\
				re^{-i \beta}  &  -t&0 \\
				0 & 0 & 1 \\
			\end{array}
			\right]
			\\
			\B(T)_{13}(r,\beta)_3=\left[
			\begin{array}{ccc}
				t &  0 &re^{i \beta }\\
				0 &  1&0\\
				re^{-i \beta}  & 0& -t \\ 
			\end{array}
			\right]
		\end{align}
	\end{subequations}
	
	Their implementation consists of a BS0, two  phase shifters with
	opposite phases, and $N-2$ identity connections, as illustrated
	in Fig.~\ref{FF176} for $N=3$.

	\begin{figure}[!t]		
		\includegraphics[scale=0.6]{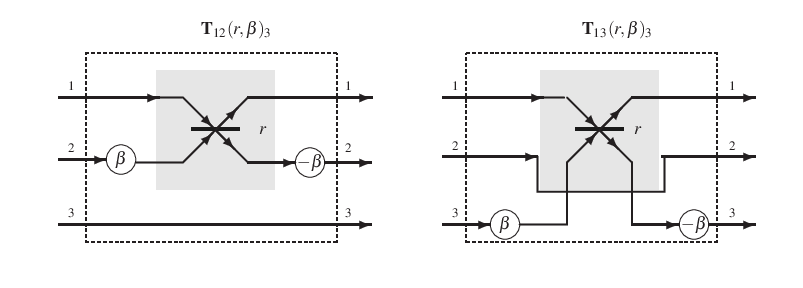}
		\centering
		\caption{Implementation of the blocks  $\B(T)_{12}(r,\beta)_3$ 
			and $\B(T)_{13}(r,\beta)_3$.}
		\label{FF176}
	\end{figure}

	\subsection{The modified Murnaghan procedure}
	
	Given an $N\times N$
	unitary matrix $\B(U)$, which we write in the polar form
	\beq
	\B(U)_N=\left[
	\begin{array}{cccc}
		u_{11} e^{i \gamma _{11}} & u _{12} e^{i \gamma _{12}} & \cdots & u _{1N} e^{i \gamma _{1N}} \\
		u _{21} e^{i \gamma _{21}} & u _{22} e^{i \gamma _{22}} & \cdots & u _{2N} e^{i \gamma _{2N}} \\
		\vdots & \vdots  &  \ddots   &\vdots  \\
		u _{N1} e^{i \gamma _{N1}} & u _{N2} e^{i \gamma _{N2}} & \cdots & u _{NN} e^{i \gamma _{NN}} \\
	\end{array}
	\right],
	\eeq
	the basic idea of the reduction procedure is to find a suitable unitary matrix $\B(V)_N$ such 
	that
	\beq
	\B(U)_N\B(V)_N=\qmatrix{ w&0\cr0&\B(U)_{N-1}\cr}\;,\label{UV}
	\eeq
	where $w$ is a complex number\ and $\B(U)_{N-1}$ is an $(N-1)\times(N-1)$ matrix. Provided that 
	$\B(U)_N$ and $\B(V)_N$ are unitary, the same holds for the right side of (\ref{UV}), so that $w$
	has modulus 1 and $\B(U)_{N-1}$ is unitary. 
	\begin{nproposition}
		The matrix $\B(V)_N$ with the desired  property (\ref{UV})
		is given by
		\beq
		\B(V)_N=\B(T)^*_{12}(r_2,\beta_2)_N\,\B(T)^*_{13}(r_3,\beta_3)_N\,\cdots\, \B(T)^*_{1N}(r_{N-1},\beta_{N-1})_N
		\label{UV3}
		\eeq
		where the parameters of the complex BSs ($N$BS$\gamma$) are  given by
		\beq
			r_{i}=\frac{u _{1i}}
			{\sqrt{u _{11}^2+\cdots+u _{1i-1}^2+u _{1i}^2}}\vq \beta_{i}=\gamma _{1i}-\gamma _{11}
		\label{H4n}
		\eeq
		where $i=2,\ldots, N$.
		The complex number $w$ is given by
		\beq
		w=e^{i\gamma_{11}}
		\label{LnA}
		\eeq
	\end{nproposition}

	It is important to note that the reduction of the unitary matrix from the order $N$ to the order $N-1$ is obtained with $N-1$ $N$BS$\gamma$, that is, with simple BSs and phase shifters.
	The reduction procedure can be applied to
	the matrix $\B(U)_{N-1}$ to get a  matrix $\B(U)_{N-2}$ of order $N-2$ and it can be repeated until one gets a matrix $\B(U)_2$ of order 2.
	This iterative procedure will be explicitly applied in the next subsection for an arbitrary unitary matrix of order 4 and finally to the matrix of the 4--DFT. 
	
	The final complexity  is \cite{Cari18}
	\beq \label{PP2}
	\mett N(N-1)\q\hb{BS0}\vq N(N-1)+1\q \hb{phase shifters}. 
	\eeq

	\subsection{The iterative  procedure for $N=4$} 
	
	We illustrate the iterative procedure for $N=4$, where the unitary matrix is
	\beq
	\B(U)_4=\qmatrix{
		u _{11}e^{i\gamma _{11}}&u_{12}e^{i\gamma_{12}}&u_{13}e^{i\gamma_{13}}&u_{14}
		e^{i\gamma_{14}}\cr
		u _{21}e^{i\gamma _{21}}&u_{22}e^{i\gamma_{22}}&u_{23}e^{i\gamma_{23}}&u_{24}
		e^{i\gamma_{24}}\cr
		u _{31}e^{i\gamma _{31}}&u_{32}e^{i\gamma_{32}}&u_{33}e^{i\gamma_{33}}&u_{34}
		e^{i\gamma_{34}}\cr
		u _{41}e^{i\gamma _{41}}&u_{42}e^{i\gamma_{42}}&u_{43}e^{i\gamma_{43}}&u_{44}
		e^{i\gamma_{44}}\cr}
	\eeq
	
	Then the reduction is performed in $N-2=2$ steps and leads to the architecture illustrated in Fig.~\ref{DE175}.

	\begin{figure}[!t]
		\includegraphics[scale=0.6]{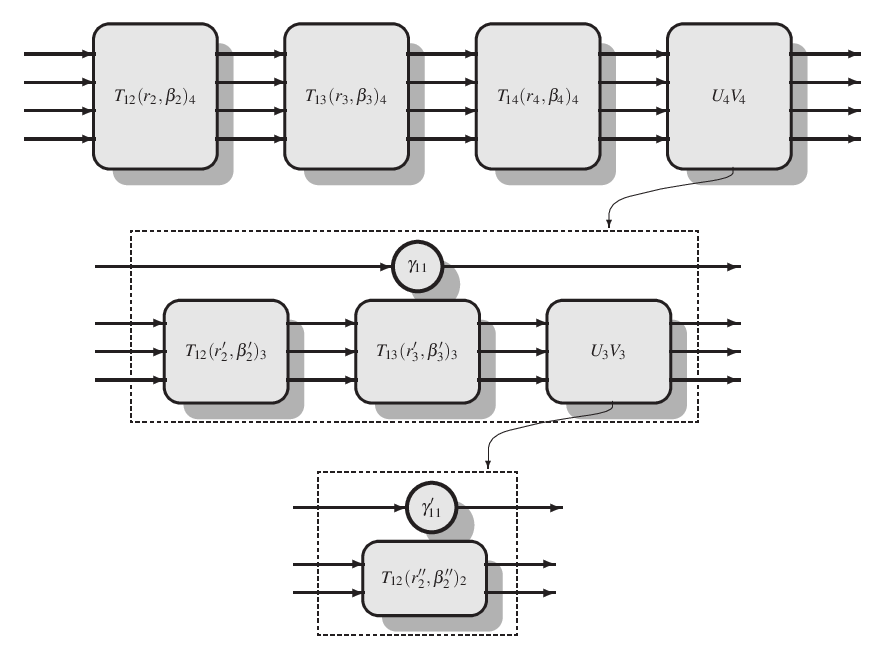}
		\centering
		\caption{Implementation of a $4\times4$ unitary matrix 
		in the general case through phase shifters and beam splitters 
		using the modified Murnaghan procedure.}
		\label{DE175}
	\end{figure}

	In {\bf Step 1} we evaluate the parameters of the 3 complex BSs in
	\beq
	\B(V)_4= \B(T)^*_{12}(r_2,\beta_2)_4\B(T)^*_{13}(r_3,\beta_3)_4\B(T)^*_{14}(r_4,\beta_4)_4
	\eeq
	given by
	\beq
	\eqalign{
		&\beta_{2}=\gamma_{12}-\gamma_{11}\quad,\quad r_{2}={u_{12}\over\sqrt{u_{11}^2+
				u_{12}^2}}\cr
		&\beta_{3}=\gamma_{13}-\gamma_{11}\quad,\quad
		r_{3}=
		{u_{13}\over\sqrt{u_{11}^2+u_{12}^2+u_{13}^2}
		}\cr
		&\beta_{4}=\gamma_{14}-\gamma_{11}\quad,\quad r_{4}=
		\frac{u_{14}}{\sqrt{u_{11}^2+u_{12}^2+u_{13}^2+u_{14}^2}}\cr}
	\label{TT4}
	\eeq
	Then 
	\beq
	\B(U)_4\B(V)_4\,
	=\qmatrix{ e^{i\gamma_{11}}&0\cr0&\B(U)_3\cr}
	\eeq

	In {\bf Step 2} we perform the reduction of the matrix $\B(U)_3$, which we write in the  
	modulus--argument form $\B(U)_3=[u'_{ij}\,e^{i \gamma'_{ij}}]\;,\;i,j=1,2,3$.
	Then 
	we evaluate the parameters of the  2 complex BSs of the middle part of 
	Fig.~\ref{DE175}, which gives
	\beq
	\B(V)_3=\B(T)^*_{23}(r'_{2},\beta'_{2})_3\B(T)^*_{24}(r'_{3},\beta'_{3})_3
	\eeq
	where
	\beq
	\eqalign{
		&\beta'_{2}=\gamma'_{12}-\gamma'_{11}\quad,\quad r'_{2}
		=\frac{u'_{12}}{\sqrt{{u'}^{2}_{11}+{u'}^{2}_{12}}}\cr
		&\beta'_{3}=\gamma'_{13}-\gamma'_{11}\quad,\quad
		r'_{3}=\frac{{u'}_{13}}{\sqrt{{u'}_{11}^2+
				{u'}_{12}^2+{u'}_{13}^2}}
		\cr}
	\label{AA5T}
	\eeq
	At this point we find
	\beq
	\B(U)_3\B(V)_3
	=\qmatrix{ e^{i\gamma'_{11}} & 0\cr
		0                 & \B(U)_2 \cr} \label{C6}
	\eeq
	Finally the unitary matrix $\B(U)_2$ of order 2 is implemented
	according to Prop.~\ref{P2}.

	\subsection{Application of the Murnaghan procedure to the to 4--cvQFT}
	
	In this section we apply the Murnaghan procedure to the 4--cvQFT. This case is sufficient to preview how the procedure works in the general case of
	$N$--cvQFT.
	
	The 4--DFT matrix
	\beq
	\B(U)_4=\B(W)_4=\met\left[
	\begin{array}{cccc}
		1 & 1 & 1 & 1 \\
		1 & i & -1 & -i \\
		1 & -1 & 1 & -1 \\
		1 & -i & -1 & i \\
	\end{array}
	\right]
	\eeq
	is unitary and can be decomposed with the general procedure in two steps.  
	In the first step,
	\beq
	\B(V)_4= \B(T)^*_{12}(r_2,\beta_2)_4\B(T)^*_{13}(r_3,\beta_3)_4\B(T)^*_{14}(r_3,\beta_4)_4
	\eeq
	where
	\beq
	r_2=\frac{1}{\sqrt{2}},\; \beta_2=0,\;
	r_3=\frac{1}{\sqrt{3}},\; \beta_3=0,\;
	r_4=\met,\; \beta_4=0
	\eeq
	Then
	\beq
	\B(U)_4\B(V)_4=\left[
	\begin{array}{cccc}
		1 & 0 & 0 & 0 \\
		0 & \frac{-1+i}{2\sqrt{2}} & -\frac{3+i}{2\sqrt{6}} & -\frac{i}{\sqrt{3}}\\
		0 & -\frac{1}{\sqrt{2}} & \frac{1}{\sqrt{6}} & -\frac{1}{\sqrt{3}}\\
		0 & -\frac{1+i}{2\sqrt{2}} & -\frac{3-i}{2\sqrt{6}} & \frac{i}{\sqrt{3}}\\
	\end{array}\right]=\qmatrix{1&0\cr0&\B(U)_3\cr}
	\eeq

	In the second step, we reduce the matrix
	\beq
	\B(U)_3=\left[
	\begin{array}{ccc}
		\frac{-1+i}{2\sqrt{2}} & -\frac{3+i}{2\sqrt{6}} & -\frac{i}{\sqrt{3}}\\
		-\frac{1}{\sqrt{2}} & \frac{1}{\sqrt{6}} & -\frac{1}{\sqrt{3}}\\
		-\frac{1+i}{2\sqrt{2}} & -\frac{3-i}{2\sqrt{6}} & \frac{i}{\sqrt{3}}\\
	\end{array}\right]
	\eeq
	by the application of 
	\beq
	\B(V)_3=\B(T)^*_{12}(r'_{2},\beta'_{2})_3\B(T)^*_{13}(r'_{3},\beta'_{3})_3
	\eeq
	where
	\begin{equation}
		\begin{split}
			r_2'&=\frac{\sqrt{10}}{4}\vq
			\beta_2'=\tan ^{-1}\left(\frac{1}{3}\right)+\frac{ \pi }{4}\\
			r_3'&=\frac{1}{\sqrt{3}}\vq \beta_3'=\frac{3 \pi }{4}
		\end{split}
	\end{equation}
	one gets
	\beq
	\begin{split}
		\B(U)_3\B(V)_3=
		\left[
		\begin{array}{ccc}
			\frac{-1+i}{\sqrt{2}} & 0 & 0 \\
			0 & \frac{1}{2}+\frac{i}{2} & -\frac{1}{\sqrt{2}} \\
			0 & -\frac{1}{2}+\frac{i}{2} & \frac{i}{\sqrt{2}} \\
		\end{array}\right]=
		\qmatrix{ e^{i\gamma'_{11}}&0\cr0&\B(U)_2\cr} 
	\end{split}
	\eeq
	where 
	\beq
	\begin{split}
		\B(U)_2=
		\left[
		\begin{array}{ccc}
			\frac{1}{2}+\frac{i}{2} & -\frac{1}{\sqrt{2}} \\
			-\frac{1}{2}+\frac{i}{2} & \frac{i}{\sqrt{2}} \\
		\end{array}\right]
	\end{split}
	\eeq
	The detailed synthesis is illustrated in Fig.~\ref{FF36}.

	\begin{figure*}[!t]
		\includegraphics[scale=0.8]{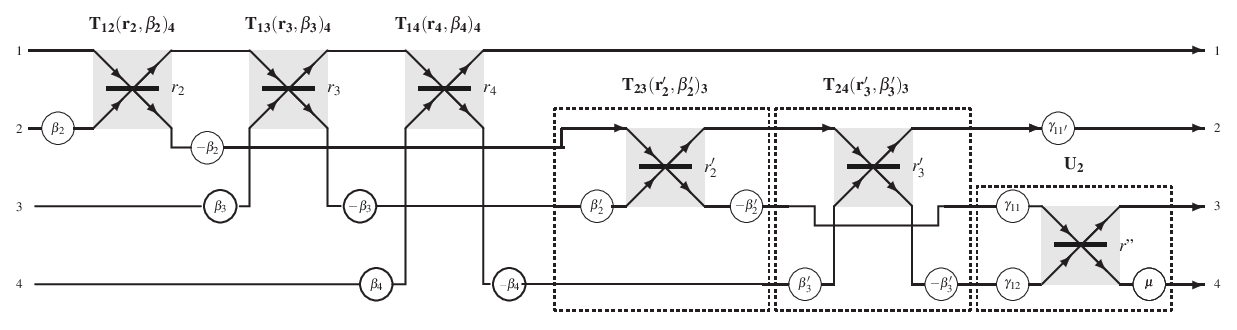}
		\centering
		\caption{Implementation of the 4--cvQFT according to the 
		modified Murnaghan approach. 
		For the values of the parameters see the text.}
		\label{FF36}
	\hrulefill
	\vspace*{0pt}
	\end{figure*}

	\section{\label{nonUT}AN EFFICIENT REDUCTION FOR THE cvQFT: TIME DECIMATION}
	
	We have seen that with the available approach the implementation 
	of the cvQFT becomes complicated just for small values of $N$, as seen
	for $N=4$. Thus, for high values of $N$,  a search
	for a more efficient approach becomes mandatory. To this end we have
	investigated
	the technique of the efficient calculation of the DFT of a deterministic
	signal in the field of DSP, known as fast Fourier transform (FFT).
	As a matter of fact, the complexity  of the Fourier transform 
	of a signal with $N$ values
	through the DFT increases with the law $N^2$, while with the evaluation
	through the FFT the law becomes $N\log N$, with revolutionary consequences.
	
	Now, following the theory of the DFT, called time decimation,
	we have found a very efficient algorithm for the implementation
	of the cvQFT. Here, we do not introduce the time decimation, 
	but we limit ourselves in the formulation of the  algorithm and we will give 
	an autonomous proof, not related to the digital signal processing.
	
	Note that there is a one--to-one correspondence between the $N$--cvQFT 
	and the $N$--DFT, so that the implementation of the  $N$--cvQFT can be
	obtained from the implementation of the $N$--DFT matrix.
	
	We consider the DFT matrix of order $N$ with $N$ a power of 2
	\beq
	\B(W)_N= \left[w_{rs}\right]_{r,s=0,1,\ldots,N-1} \with 
	w_{rs}=\frac1{\sqrt N} e^{i  2\pi rs/N}
	\label{B1A}
	\eeq	
	The fast reduction consists in the
	decomposition of the DFT matrix $\B(W)_N$ into two DFT matrices  $\B(W)_L$.
	If $N=2^m$ is a power of two, in $m-1$ iterations one can decompose
	the original matrix $\B(W)_N$ into DFT matrices $\B(W)_2$.

	\begin{ntheorem}
		Let $N$ be an arbitrary even integer and let $L=\frac{N}{2}$.
		Then the $N$--cvQFT can be implemented by the following steps:
		\IT 1) Split the input modes
		\beq
		\B(a)=[\hat{a}_0,\hat{a}_1,\ldots, \hat{a}_{N-1}]^\TT
		\label{B2}
		\eeq
		into the two subsets $\B(a)_0$ and $\B(a)_1$ of size $L=\frac{N}{2}$ 
		\beq
		\B(a)_0=[\hat{a}_0,\hat{a}_2,\ldots,\hat{a}_{N-2}]^\TT\vq \B(a)_1=[\hat{a}_1,\hat{a}_3,\ldots,\hat{a}_{N-1}]^\TT.
		\label{B3}
		\eeq
		\IT 2) Two $L$--point cvQFT, giving
		\begin{equation}
			\begin{split}
				\hat{b}_{0k}&= \frac1{\sqrt L}\sum_{j=0}^{L-1}\; \hat{a}_{2j}\; e^{i 2\pi{kj}/L} \\
				\hat{b}_{1k}&= \frac1{\sqrt L}\sum_{j=0}^{L-1}\; \hat{a}_{2j+1}\; e^{i 2\pi{kj}/L}
			\end{split}
		\end{equation}
		
		\IT 3) A phase shift of the components of the second subset 
		by $w_N^{k}$, \ $k=0,1,\ldots,L-1$.
		
		\IT 4) N parallel 2-cvQFT (beam-splitter) on the $k$th modes of the subsets, gives
		\begin{equation}
			\begin{split}
				\hat{A}_{k} &= \hat{A}_{k}^+= \frac1{\sqrt 2}(\hat{b}_{0k} + \hat{b}_{1k}) \\
				&=\frac1{\sqrt N} \sum_{j=0}^{L-1}\; (\hat{a}_{2j}\; e^{i 2\pi{k(2j)}/N} + \hat{a}_{2j+1}\; e^{i 2\pi{k(2j+1)}/N})\\
				&=\frac1{\sqrt N} \sum_{j=0}^{N-1}\; \hat{a}_{j}\; e^{i 2\pi{kj}/N}
			\end{split}
		\end{equation}
		\begin{equation}
			\begin{split}
				\hat{A}_{L+k}&= \hat{A}_{k}^- = \frac1{\sqrt 2}(\hat{b}_{0k} - \hat{b}_{1k}) \\
				&=\frac1{\sqrt N} \sum_{j=0}^{L-1}\; (\hat{a}_{2j}\; e^{i 2\pi{k(2j)}/N} - \hat{a}_{2j+1}\; e^{i 2\pi{k(2j+1)}/N})\\
				&=\frac1{\sqrt N} \sum_{j=0}^{L-1}\; (\hat{a}_{2j}\; e^{i 2\pi{(L+k)(2j)}/N} + \hat{a}_{2j+1}\; e^{i 2\pi{(L+k)(2j+1)}/N})\\
				&=\frac1{\sqrt N} \sum_{j=0}^{N-1}\; \hat{a}_{j}\; e^{i 2\pi{(L+k)j}/N}\\
			\end{split}
		\end{equation}
		
		Hence, The final annihilation mode $$\B(A)=[\hat{A}_{0}, \hat{A}_{1}, \ldots,\hat{A}_{N-1}]$$ provides the $N$-cvQFT of the modes $\B(a)$,
		\beq
		\B(A)=\B(W)_N\;\B(a)
		\eeq
		\label{DFT}
	\end{ntheorem}
	The procedure is illustrated in Fig.~\ref{FF43}
	where the cvQFT of order $N=8$, denoted as $\C(F)_8$ is decomposed into two cvQFTs of order $L=4$, denoted as $\C(F)_4$.

	\begin{figure*}[!h]
		\includegraphics[scale=1.0]{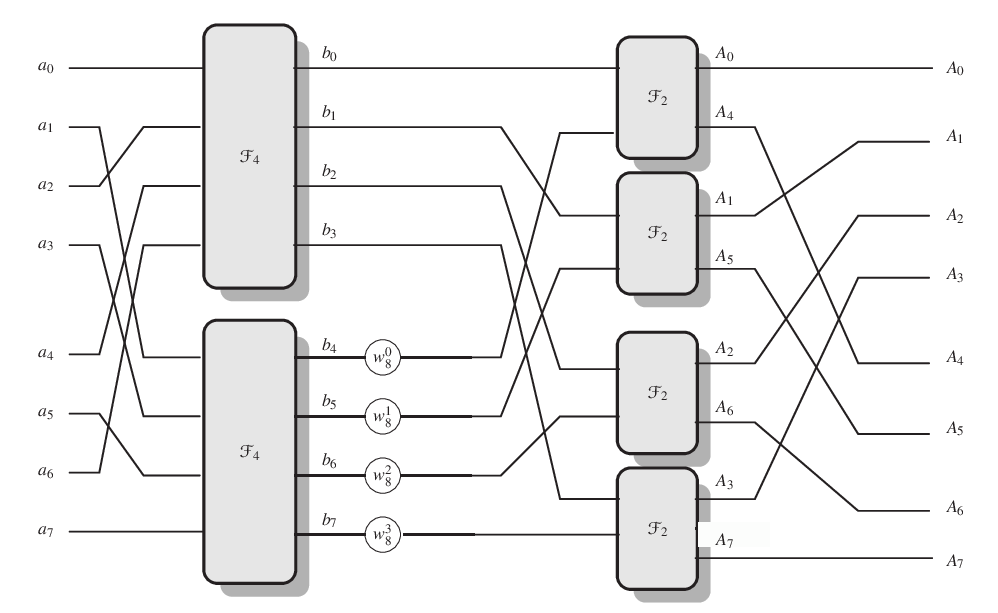}
		\centering
		\caption{Implementation of the 8--cvQFT through 2 cvQFTs of 
		order 4. There are several mode permutations 
		(changing the order of modes), which have no 
		computational complexity.}
		\label{FF43}
	\hrulefill
	\vspace*{0pt}
	\end{figure*}

	\subsection{Iterations of the fast reduction}
	
	The reduction procedure can be iterated. For a given order $N=2^m$, the first 
	iteration gives $\C(F)_N$  expressed through two $\C(F)_{N/2}$, in the second
	iteration the two $\C(F)_{N/2}$ are expressed through four $\C(F)_{N/4}$, 
	and so on.  Finally, at step $m-1$,  the original $\C(F)_N$ is expressed
	through the $\C(F)_2$  DFTs.

	In the final architecture there are several permutations of the connections, but the numerical complexity is confined to the DFTs of order 2, denoted by $\C(F)_2$, and to the
	phase rotations.
	The $\C(F)_2$  matrix is
	\beq
	\B(W)_2=\frac{1}{\sqrt2}\qmatrix{1 &1\cr1&-1\cr}
	\label{H6}
	\eeq
	According to proposition \ref{P2}, it can be implemented by a single beam splitter as in Fig.~\ref{FF278}.
	\begin{figure}[!t]
		\includegraphics[scale=0.9]{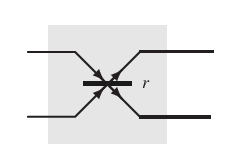}
		\centering
		\caption{Implementation of the 2--cvQFT.}
		\label{FF278}
	\end{figure}

	The global complexity  of the cvQFT of order $N=2^m$ is
	
	\beq
	\frac{N}{2} \log_2(N)\hb{ beam splitters}, \quad \frac{N}{2}  \log_2 \left( \frac{N}{2} \right) \hb{ phase shifters} \,.
	\label{PP4}
	\eeq
	
	In fact the number of BSs is equal to the number of $\C(F)_2$.
	Denoting by $T_N$ the number of beam splitters and phase shifters with the order $N=2^m$, we have the recurrence
	\beq
	T_N=2  T_{N/2}+\frac{N}{2}  \vq N=4,8,16,\ldots
	\eeq
	with $T_2=1$ for the beam splitter and $T_2=0$ for the phase shifter. The solution is indicated in expression.~(\ref{PP4}).  This result should be compared with (\ref{PP2}) related to the Murnaghan  procedure.

	\subsection{Fast implementation for $N=4$}
	In  Fig.~\ref{FF35} a detailed synthesis  of the 4--cvQFT is shown. The comparison with fig.~\ref{FF36} shows the complexity reduction achieved with the fast procedure.

	\begin{figure*}[!h]
		\includegraphics[scale=0.9]{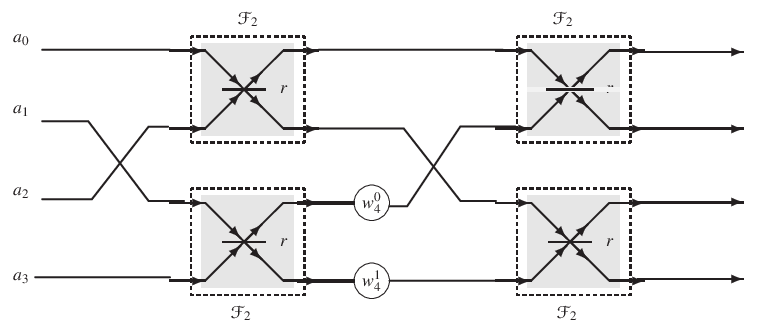}
		\centering
		\caption{Fast implementation of the 4--cvQFT through
		 $\frac{1}{2}N \log (N) = 4$ beam splitters and $\frac{1}{2}N 
		 \log \left( \frac{N}{2} \right) = 2$ phase shifters. Also,
		  there are several mode permutations, 
		  which have no computational complexity.}
	\label{FF35}
	\hrulefill
	\vspace*{0pt}
	\end{figure*}

	\section*{Part II: Applications of cvQFT}

	In this part, the cvQFT will be applied to Gaussian unitaries and to Gaussian states and therefore their formulation is needed. We introduce the main specifications.

	\section{\label{GU} GAUSSIAN UNITARIES AND THE cvQFT}

	\subsection{Gaussian unitaries in the bosonic Hilbert space}

	The Gaussian unitaries can be specified in terms of the cascade
	combination of three fundamental Gaussian unitaries (FGUs).
	The three FGUs 
	are  defined by the following unitary operators, expressed  in terms of the column
	vectors $\B(a)_*$ and  $\B(a)$ of the bosonic operators $a_i^*$ and $a_i$.
	\IT 1) { $N$--mode displacement operator}
	\beq
	D(\BB(\alpha)):=\E^{ \BB(\alpha)^\TT \B(a)_*\,\,-\BB(\alpha)^*\, \B(a)} 
	,\;\;\BB(\alpha)=[\alpha_1,\ldots,\alpha_N]^\TT \,\in \M(C)^N
	\label{U22}
	\eeq
	which is the same as the Weyl operator.
	\IT 2) { $N$--mode rotation operator}
	\beq
	R(\BB(\phi)):=\E^{\,\I\,\, \B(a)^*\BB(\phi) \, \B(a)},\;\;\BB(\phi) 
	\;\;\hb{is a $N\times N$ Hermitian matrix}.
	\label{U24}
	\eeq
	\IT 3) { $N$--mode squeeze operator}
	\beq
	S(\B(z)):=\E^{\met\left[\, \B(a)^* \,\B(z)\, \B(a)_*-\B(a)^\TT\,\B(z)^*\,\B(a)\right]} 
	,\;\;\B(z)\; \hb{is a $N\times N$ symmetric matrix}.
	\label{U26}
	\eeq
	\ET 
	Combination of these operators allows us to get all 
	the Gaussian unitaries. In fact:

	\begin{ntheorem}
		The most general Gaussian unitary is given by the combination of the
		three fundamental Gaussian unitaries $D(\alpha)$, $S(z)$,
		and $R(\phi)$, cascaded in any arbitrary order, that is, 
		$
		S(z)\,D(\alpha)\,R(\phi)\vl
		R(\phi)\,D(\alpha)\,S(z)\vl\hb{etc.}
		$
		\label{P5}
	\end{ntheorem}

	\ET This important theorem was proved by Ma and Rhodes \cite{MaRh90} using Lie's algebra.
	
	Although the FGUs act on a infinite dimensional Hilbert space, they
	are completely specified by finite dimensional parameters: 
	the displacement operator by the displacement  vector $\BB(\alpha)$,
	the rotation operator by
	the rotation matrix $\BB(\phi)$, and the squeeze operator 
	by the  squeeze matrix $\B(z)$. In the manipulations the squeeze matrix,
	which is complex symmetric, must be decomposed in the polar form \cite{Horn}
	$\B(z)=\B(r)\,\E^{\,\I\, \BB(\theta)}$,
	where $\B(r)$  is Hermitian positive semidefinite (PSD) and $\BB(\theta)$ is Hermitian
	and symmetric.

	Note that in a cascade combination one can switch the
	order of operators with appropriate change in the parameters (switching rules):
	\begin{align}
		S\!\big(\B(z)\big)\,R\!\big(\BB(\phi)\big)
		&= R\!\big(\BB(\phi)\big)\,S\!\big(\B(z)_0\big),
		& \B(z) &= \E^{\I \BB(\phi)}\,\B(z)_0\,
		\E^{\I \BB(\phi)^{\TT}}
		\label{SR1b}
		\\[6pt]
		D\!\big(\BB(\alpha)\big)\,R\!\big(\BB(\phi)\big)
		&= R\!\big(\BB(\phi)\big)\,D\!\big(\BB(\beta)\big),
		& \BB(\alpha) &= \E^{\I \BB(\phi)}\,\BB(\beta)
		\label{SR1c}
		\\[6pt]
		R\!\big(\BB(\theta)\big)\,R\!\big(\BB(\phi)\big)
		&= R\!\big(\BB(\phi)\big)\,R\!\big(\BB(\theta')\big),\quad
		& \BB(\theta') &= \E^{-\I\BB(\phi)}\,\B(\theta)\,
		\E^{\I\BB(\phi)}
		\label{SR1d}
	\end{align}

	The problem is the evaluation of the Bogoliubov matrices in terms of the FGU parameters. For the cascade 
	$D(\BB(\alpha))R(\BB(\phi))S(\B(z))$ shown in Fig.~\ref{FF137}
	The Bogoliubov matrices are given by \cite{MaRh90,CarQC}
	\beq
	\B(E)=\cosh (\B(r))\;\E^{\I\BB(\phi)}\vq \B(F)=\sinh (\B(r))
	\E^{\I\BB(\theta)}\E^{\I\BB(\phi)^\TT}
	\eeq
	With the application of the cvQFT, we have to add the operator
	$R(\BB(\phi)_{\hbt{DFT}})$ 
	at the end of the cascade. The
	switching rule (\ref{SR1c}) allows us to move the cvQFT operator
	before the displacement by modifying
	the displacement vector $\BB(\alpha)$ as
	\beq
	\BB(\alpha)_{\hbt{QFT}}=\E^{\I\phi_{\hbt{DFT}}}\,\BB(\alpha) = \B(W)_N \,\BB(\alpha)
	\label{TT7}
	\eeq
	Then, the switching rule (\ref{SR1d}) allows us to move the cvQFT operator before the rotation by modifying
	the rotation matrix $\BB(\phi)$ as 
	\beq
	\BB(\phi)_{\hbt{QFT}} = \E^{-\I\phi_{\hbt{DFT}}}\,\B(\phi)\,
	\E^{\I\phi_{\hbt{DFT}}} = \B(W)_N^* \,\B(\phi)\,\B(W)_N
	\label{TT8}
	\eeq
	Consequently, the switching rule (\ref{SR1b}) allows us to move the cvQFT operator before the squeezing by modifying
	the squeeze matrix $\B(z)$ as 
	\beq
	\B(z)_{\hbt{QFT}} = \E^{\I \phi_{\hbt{DFT}}}\,\B(z)\,
	\E^{\I \phi_{\hbt{DFT}}^{\TT}} = \B(W)_N \,\B(z)\,\B(W)_N
	\label{TT8}
	\eeq
	
	In conclusion,
	\begin{align}
	\BB(\alpha)_{\hbt{QFT}}(k) & = \frac{1}{\sqrt{N}} \sum_{m=0}^{N-1}  e^{i\frac{2\pi}{N}mk} \alpha_m \\
	\BB(\phi)_{\hbt{QFT}}(k,l) & = \frac{1}{N} \sum_{m,n=0}^{N-1}  e^{i\frac{2\pi}{N}(-mk + nl)} \BB(\phi)_{m,n} \\	
	\B(z)_{\hbt{QFT}}(k,l) & = \frac{1}{N} \sum_{m,n=0}^{N-1}  e^{i\frac{2\pi}{N}(mk + nl)} \B(z)_{m,n} 
	\end{align}

	\begin{figure*}[!t]
		\includegraphics[scale=0.9]{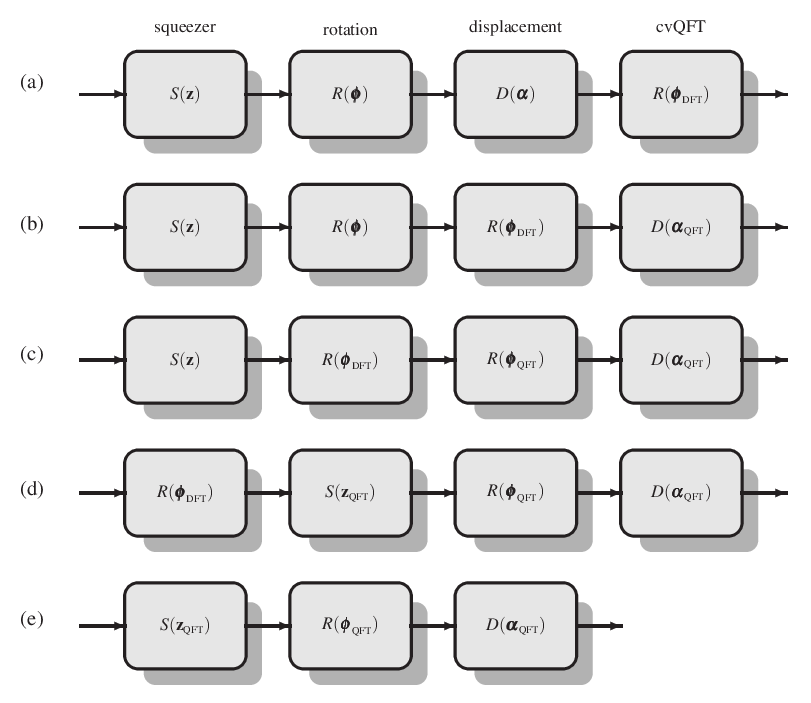}
		\centering
		\caption{(a) Application of the cvQFT  after the 
		cascade of FGUs. 
		(b) The switching rule allows the inversion of the  
		displacement and of the cvQFT. 
		(c) The switching rule allows the inversion of 
		the rotation and of the cvQFT. 
		(d) The switching rule allows the inversion of the 
		squeezing and of the cvQFT. 
		(e) Remove cvQFT rotation for its irrelevance.}
	\label{FF137}
	\end{figure*}

	\begin{nproposition}
		The application of the cvQFT to the end of the cascade of
		 Fig.~\ref{FF137} has the simple effect of modifying
		the displacement vector to its (one dimensional) discrete Fourier transform, the squeeze matrix to its (two dimensional) discrete Fourier transform, and the rotation matrix to a Fourier like transform.
		\label{Q3}
	\end{nproposition}

	\subsection{Gaussian unitaries in the phase space}
	
	In the phase space $N$--mode Gaussian  unitaries are specified by
	the symplectic matrix.  There are two versions of symplectic matrices,
	a real version $\B(S)_r$ and complex version $\B(S)_c$  both of order $2N$,
	which verify the symplectic condition
	\beq
	\B(S)_r\BB(\Omega)\B(S)_r^\TT=\BB(\Omega)\vq 
	\B(S)_c\BB(\Omega)\B(S)_c^*=\BB(\Omega)
	\with \BB(\Omega)=\qmatrix{\B(0)&\B(I)\cr-\B(I)&\B(0)\cr}
	\eeq
	where $\B(I)$ is the unitary matrix.
	Here we prefer the complex version because it is simply related to
	Bogoliubov matrices, specifically \cite{Adesso}
	\beq
	\B(S)_c=\qmatrix{\B(E)&\B(F)\cr \ov{\B(F)}&\ov{\B(E)}\cr}
	\label{R3}
	\eeq
	In particular for the cvQFT, where $\B(E)=\B(W)_N$ and
	$\B(F)=\B(0)$, we find the block diagonal form
	\beq
	\B(S)_W=\qmatrix{\B(W)_N&\B(0)\cr\B(0)&\ov{\B(W)}_N\cr}
	\label{R4}
	\eeq
	Now it is easy to find the effect of the cvQFT on the symplectic matrix,
	namely
	\beq
		\B(S)_c^{\hbt{QFT}}=\B(S)_W\,\B(S)_c
	\label{R6}
	\eeq
	
	\subsection{Example of application}
	
	We consider as an example of application a Gaussian
	unitary related to a Gaussian state discussed by several authors 
	\cite{vanLoock2002,vanLoock2003,Giedke2001} in the context of 
	continuous pure states with
	interesting forms of entanglement. In the cited papers general $N$ 
	mode states are considered. As a particular case we
	consider a four-mode state generated by a Gaussian unitary 
	characterized by Bogoliubov matrices
	\beq
	\B(E)=\left[
	\begin{array}{cccc}
		u & v & v & v \\
		v & u & v & v \\
		v & v & u & v \\
		v & v & v & u \\
	\end{array}\right]
	\vq
	\B(F)=\left[
	\begin{array}{cccc}
		x & y & y & y \\
		y & x & y & y \\
		y & y & x & y \\
		y & y & y & x \\
	\end{array}\right]
	\eeq
	where
	\beq
	u=\frac{1}{4} (c_1+3 c_2)\vq v=\frac{1}{4} (c_1-c_2)
	\eeq
	\beq
	x=\frac{1}{4} (s_1-3 s_2)\vq y=\frac{1}{4} (s_1+s_2)
	\eeq
	and $c_i = \cosh(r_i)$ and $s_i = \sinh(r_i)$.
	The authors do not give the expression of the squeeze matrix.
	We find
	\beq
	\B(r)= \frac{1}{4}\left[
	\begin{array}{cccc}
		r_1+3 r_2 &r_1-r_2 &r_1-r_2 &r_1-r_2 \\
		r_1-r_2 &r_1+3 r_2 &r_1-r_2 &r_1-r_2 \\
		r_1-r_2 &r_1-r_2 &r_1+3 r_2 &r_1-r_2 \\
		r_1-r_2 &r_1-r_2 &r_1-r_2 &r_1+3 r_2 \\
	\end{array}\right]
	\vq\eeq\beq
	e^{i\BB(\theta)}=\met\left[
	\begin{array}{cccc}
		-1 & 1 & 1 & 1 \\
		1 & -1 & 1 & 1 \\
		1 & 1 & -1 & 1 \\
		1 & 1 & 1 & -1 \\
	\end{array}\right] 
	\eeq
	\beq
	\B(z)=\B(r)\;e^{i\BB(\theta)}=\frac14\left[
	\begin{array}{cccc}
		r_1-3 r_2 &r_1+r_2 &r_1+r_2 &r_1+r_2 \\
		r_1+r_2 &r_1-3 r_2 &r_1+r_2 &r_1+r_2 \\
		r_1+r_2 &r_1+r_2 &r_1-3 r_2 &r_1+r_2 \\
		r_1+r_2 &r_1+r_2 &r_1+r_2 &r_1-3 r_2 \\
	\end{array}\right] 
	\label{R70}
	\eeq
	The complex symplectic matrix is given by Eq. (\ref{example1}) and it is modified by the cvQFT as Eq. (\ref{example2}).

	\begin{figure*}[!t]
		\beq \label{example1}
		%\hskip-15mm
		\B(S)_c=\frac14
		\scalebox{1.0}{$
			\left[
			\begin{array}{cccccccc}
				c_1+3 c_2& c_1-c_2& c_1-c_2& c_1-c_2& s_1-3 s_2 & s_1+s_2 & s_1+s_2 & s_1+s_2 \\
				c_1-c_2& c_1+3 c_2& c_1-c_2& c_1-c_2& s_1+s_2 & s_1-3 s_2 & s_1+s_2 & s_1+s_2 \\
				c_1-c_2& c_1-c_2& c_1+3 c_2& c_1-c_2& s_1+s_2 & s_1+s_2 & s_1-3 s_2 & s_1+s_2 \\
				c_1-c_2& c_1-c_2& c_1-c_2& c_1+3 c_2& s_1+s_2 & s_1+s_2 & s_1+s_2 & s_1-3 s_2 \\
				s_1-3 s_2 & s_1+s_2 & s_1+s_2 & s_1+s_2 & c_1+3 c_2& c_1-c_2& c_1-c_2& c_1-c_2\\
				s_1+s_2 & s_1-3 s_2 & s_1+s_2 & s_1+s_2 & c_1-c_2& c_1+3 c_2& c_1-c_2& c_1-c_2\\
				s_1+s_2 & s_1+s_2 & s_1-3 s_2 & s_1+s_2 & c_1-c_2& c_1-c_2& c_1+3 c_2& c_1-c_2\\
				s_1+s_2 & s_1+s_2 & s_1+s_2 & s_1-3 s_2 & c_1-c_2& c_1-c_2& c_1-c_2& c_1+3 c_2\\
			\end{array}\right]
			$}
		\eeq
		
		\hrulefill
		\vspace*{0pt}
	\end{figure*}

	\begin{figure*}[!t]
		\beq \label{example2}
		\B(S)_c^{\hbt{QFT}} =\B(S)_W\,\B(S)_c =
		\scalebox{1.0}{$
			\left[
			\begin{array}{cccccccc}
				c_1 & c_1 & c_1 & c_1 & s_1 & s_1 & s_1 & s_1 \\
				c_2 & i c_2 & -c_2 & -i c_2 & -s_2 & -i s_2 & s_2 & i s_2 \\
				c_2 & -c_2 & c_2 & -c_2 & -s_2 & s_2 & -s_2 & s_2 \\
				c_2 & -i c_2 & -c_2 & i c_2 & -s_2 & i s_2 & s_2 & -i s_2 \\
				s_1 & s_1 & s_1 & s_1 & c_1 & c_1 & c_1 & c_1 \\
				-s_2 & i s_2 & s_2 & -i s_2 & c_2 & -i c_2 & -c_2 & i c_2 \\
				-s_2 & s_2 & -s_2 & s_2 & c_2 & -c_2 & c_2 & -c_2 \\
				-s_2 & -i s_2 & s_2 & i s_2 & c_2 & i c_2 & -c_2 & -i c_2 \\
			\end{array}\right]
			$}
		\eeq
		
		\hrulefill
		\vspace*{0pt}
	\end{figure*}

	\section{\label{BG} GAUSSIAN STATES AND THE cvQFT}

	Gaussian unitaries  applied to ground state provide pure Gaussian states
	and applied to thermal states provides mixed states.
	
	\subsection{Gaussian states in the bosonic Hilbert space}

	\evv{1. \q Effect of the cvQFT on pure Gaussian states}
	
	Theorem \ref{P5} states that the most general Gaussian unitary is given by the
	combination of the three fundamental Gaussian unitaries. For pure Gaussian
	states we have:

	\begin{ntheorem}  	
		The most general $N$--mode pure Gaussian state is obtained from
		the $N$ replica of the vacuum state $\ket{\B(0)}$ as
		$$ 
		\ket{\BB(\alpha),\B(z)}:=D(\BB(\alpha))S(\B(z))\ket{\B(0)}  \label{BB2}
		$$
		\label{T2}
	\end{ntheorem} 
	
	The reason of the absence of the rotation in theorem~\ref{BB2} is due to the fact
	that the application of the vacuum state to the rotation operator
	gives back the vacuum state itself, that is, 
	$R(\BB(\phi))\ket{\B(0)}=\ket{\B(0)}$. The statement legitimates
	to denote
	$\ket{\BB(\alpha),\B(z)}$ as a general pure Gaussian state and therefore the specification is confined to the $N$ vector $\BB(\alpha)$ and to the $N\times N$ symmetric  matrix $\B(z)$.
	
	With the application of the cvQFT one finds:
	
	\begin{nproposition}
		The application of the cvQFT 
		modified a pure Gaussian state as 
		\beq
		\ket{\BB(\alpha),\B(z)} \arrcvQFT
		\ket{\BB(\alpha_{\hbt{QFT}}),\B(z)_{\hbt{QFT}}}
		\label{BB4}
		\eeq
		where
		\beq
		\BB(\alpha)_{\hbt{QFT}}=e^{i\BB(\phi)_{\hbt{DFT}}}\BB(\alpha)
		\vq\B(z)_{\hbt{QFT}}=e^{i\BB(\phi)_{\hbt{DFT}}}\,\B(z)\,
		e^{i\BB(\phi)^\TT_{\hbt{DFT}}}
		\label{BBAa}
		\eeq
		\label{P88}
	\end{nproposition}

	\begin{figure}[!t]
		\includegraphics[scale=0.8]{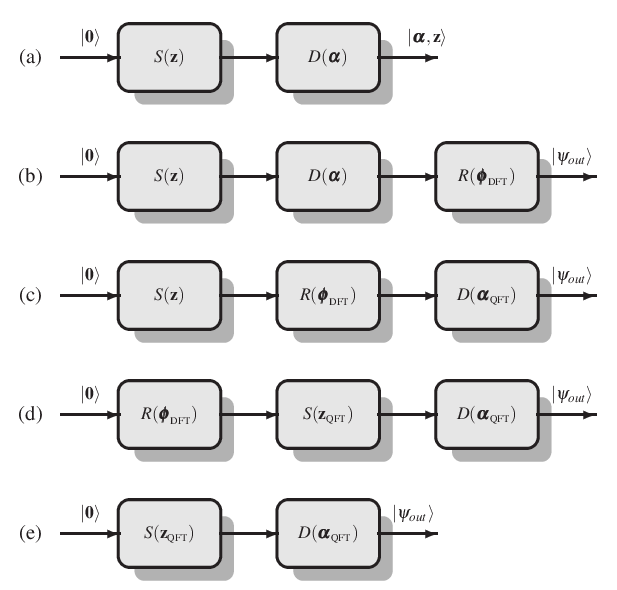}
		\centering
		\caption{(a) Generation of the Gaussian state from the vacuum 
			state $\ket{\B(0)}$.
			(b)  Introduction  of the  cvQFT  through the operator
			$R(\BB(\phi)_{\hbt{DFT}})$ with consequent modification 
			of the output state.
			(c) Inversion of displacement and rotation  with 
			modification of displacement.
			(d) Inversion of squeezing and rotation with m
			odification of squeezing.
			(e) Remove rotation for its irrelevance.}
		\label{FF138}
	\end{figure}

	We follow Fig.~\ref{FF138}. In (a) the generation of a
	standard pure Gaussian state. In (b) the introduction of the operator $R(\BB(\Phi)_{\hbt{DFT}})$ which provides  the cvQFT. In (c) the
	inversion of displacement and rotation using the switching rule of Eq. (\ref{SR1c}), thus the displacement vector becomes
	\beq
	\BB(\alpha)_{\hbt{QFT}}=e^{i \BB(\phi)_{\hbt{DFT}}}\BB(\alpha) =\B(W)_N \BB(\alpha)
	\eeq
	In (d) the inversion of squeezing and rotation using the switching rule of Eq. (\ref{SR1b}), thus the squeeze matrix  becomes (considering that
	the DFT matrix is symmetric)
	\beq
	\B(z)_{\hbt{QFT}}=e^{i\BB(\phi)_{\hbt{DFT}}}\,\B(z)\,e^{i\BB(\phi)^\TT_{\hbt{DFT}}}=
	\B(W)_N\B(z)\B(W)_N
	\eeq

	\evv{2. \q Effect on mixed Gaussian states}
	
	Williamson’s  theorem provided the generation of mixed Gaussian states starting from thermal noise \cite{CarQC}:

	\begin{ntheorem} 
		The most general $N$--mode Gaussian state is generated
		from thermal state by application of the three 
		fundamental unitaries as
		\beq
		\rho(\BB(\alpha),\BB(\phi),\B(z)|\B(V)^\oplus)= U(\BB(\alpha),\BB(\phi),\B(z))\;\rho_{\hbf{th}}(\B(V)^\oplus)\;
		U^*(\BB(\alpha),\BB(\phi),\B(z))      
		\label{W12}
		\eeq
		where
		\beq
		U(\BB(\alpha),\BB(\phi),\B(z))=
		D(\BB(\alpha))\,R(\BB(\phi))\,S(\B(z))
		\eeq
		and $\rho_{\hbf{th}}(\B(V)^\oplus)$ is an $N$--mode thermal noise. 
		\label{H4}
	\end{ntheorem}
	With the application of the cvQFT one finds:

	\begin{nproposition} 
		The application of the cvQFT 
		modified a mixed Gaussian state as 
		\beq
		\rho(\BB(\alpha),\BB(\phi),\B(z)|\B(V)^\oplus) \arrcvQFT
		\rho(\BB(\alpha)_{\hbt{QFT}},\BB(\phi)_{\hbt{QFT}},\B(z)_{\hbt{QFT}}|\B(V)^\oplus)
		\label{BB4}
		\eeq
		where
		\beq
		\BB(\alpha)_{\hbt{QFT}}=\B(W)_N\BB(\alpha)\vq
		\BB(\phi)_{\hbt{QFT}}=\B(W)_N^* \BB(\phi) \B(W)_N \vq
		\B(z)_{\hbt{QFT}} = \B(W)_N \,\B(z)\,\B(W)_N
		\eeq
		\label{P9mixed8}
	\end{nproposition}
	
	\begin{proof}
		We follow Fig.~\ref{FF139}.
		In (a) the generation of the Gaussian state from the thermal state 
		$\rho(\B(V)^{\otimes})$.
		In (b)  the introduction  of the  cvQFT  through the operator $R(\BB(\phi)_{\hbt{DFT}})$ with consequent modification of the output state.
		In (c) the inversion of the displacement with modification of displacement 
		\beq
		\BB(\alpha)_{\hbt{QFT}}=\B(W)_N\BB(\alpha)
		\eeq
		In (d) the inversion of the rotation with modification of rotation matrix  
		\beq
		\BB(\phi)_{\hbt{QFT}}=\B(W)_N^* \BB(\phi) \B(W)_N
		\eeq
		In (e) the inversion of the squeezing with modification of squeezing matrix 
		\beq
		\B(z)_{\hbt{QFT}} = \B(W)_N \,\B(z)\,\B(W)_N
		\eeq

		\begin{figure*}[!t]
			\includegraphics[scale=0.9]{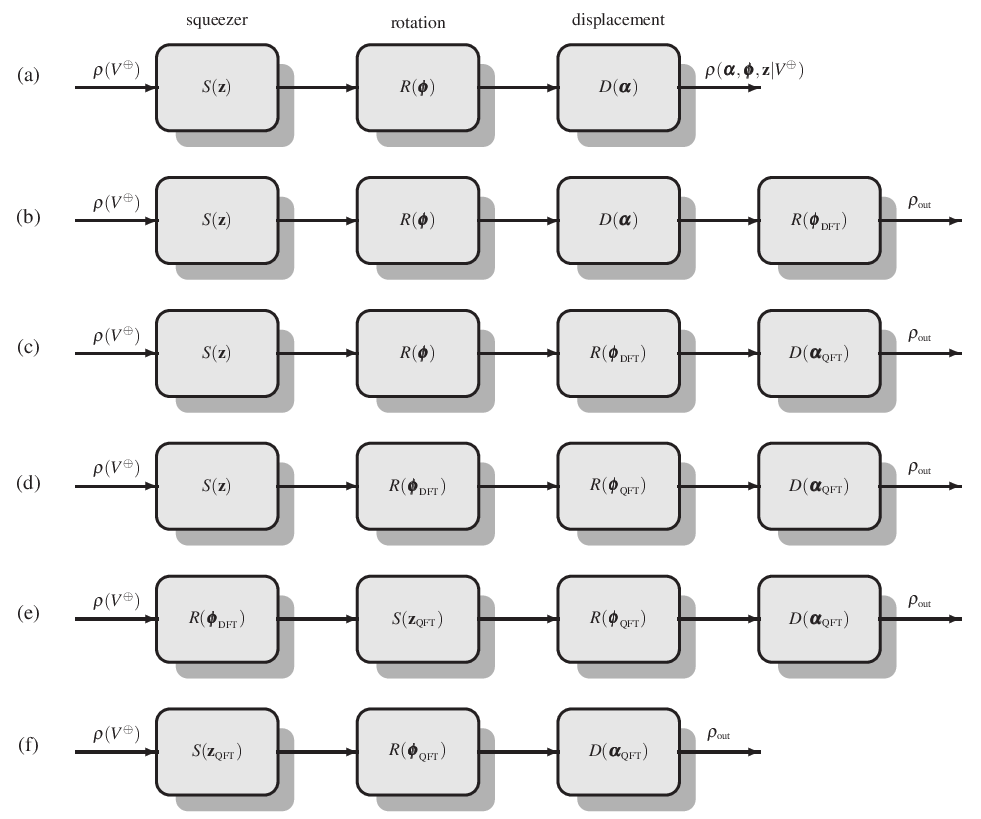}
			\centering
			\caption{(a) Generation of the Gaussian state 
			from the thermal state $\rho_{th}$.
				(b)  Introduction  of the  cvQFT  
				through the operator $R(\BB(\phi)_{\hbt{DFT}})$ 
				with consequent modification of the 
				output state.
				(c) Inversion of displacement and rotation 
				with modification of displacement.
				(d) Inversion of two rotations with 
				modification of the rotation.
				(e) Inversion of squeezing and rotation 
				with modification of squeezing.
				(f) Remove rotation for its irrelevance.}
			\label{FF139}
		\vspace*{3pt}
		\hrulefill
		\vspace*{3pt}
		\end{figure*}

	\end{proof}

	\subsection{Gaussian states in the phase space}
	In the phase space $N$--mode Gaussian states are completely
	described by the mean vector $\B(m)$
	and the covariance matrix (CM) $\B(V)$, where $\B(m)$ is a vector of size 
	$2N$  and $\B(V)$ is a matrix
	of order $2N$. Also, for the CM we may have
	a real and complex version related by the unitary matrix $L=\frac{1}{\sqrt{2}} \qmatrix{\B(I)_N&\B(I)_N\cr -i\B(I)_N&\B(I)_N\cr}$, but in
	this case we find more convenient the real version.
	
	After a transformation with (real) symplectic matrix $\B(S)_r$
	the mean value and the (real) covariance matrix become
	\beq
	\B(m) \mapsto \B(S)_r\,\B(m)\vq \B(V) \mapsto \B(S)_r\,\B(V)\,\B(S)^\TT_r
	\label{Z49}
	\eeq
	A pure Gaussian state is generated from the vacuum state $\ket{\B(0)}$
	and its CM becomes
	\beq
	\B(V)= \B(S)_r\B(V)_0\B(S)_r^\TT=\B(S)_r\B(S)_r^\TT  
	\label{R20}
	\eeq
	in consideration of the fact that the CM of the ground state $\B(V)_0$
	is the identity. A mixed  Gaussian state is generated from the thermal state $\rho_{\hbf{th}}$
	and its CM becomes
	\beq
	\B(V)= \B(S)_r\B(V)_{\hbf{th}}\B(S)_r^\TT  
	\label{R22}
	\eeq

	Thus, the effect of the cvQFT on the CM $\B(V)$ results in
	\beq
		\B(V)_{\hbt{QFT}} = \B(S)_{Wr}\B(V)\B(S)_{Wr}^\TT \label{R24}
	\eeq
	where $\B(S)_{Wr}=\B(L)\B(S)_W\B(L)^*$.
	Note that Eq. (\ref{R24}) holds for both pure and mixed Gaussian states.

	\subsection{Example of application (cont.)}
	
	It is interesting to see the effect of the cvQFT on the squeeze matrix given by Eq. (\ref{R70}). We find
	\beq
	\B(z)_{\hbt{QFT}}=\B(W)_N\B(z)\B(W)_N=\frac14\left[
	\begin{array}{cccc}
		r_1 & 0 & 0 & 0 \\
		0 & 0 & 0 & -r_2 \\
		0 & 0 & -r_2 & 0 \\
		0 & -r_2 & 0 & 0 \\
	\end{array}\right]
	\eeq
	The polar decomposition 
	$\B(z)_{\hbt{QFT}}= \B(r)_{\hbt{QFT}}e^{i\BB(\theta)_{\hbt{QFT}}}$
	gives
	\beq
	\B(r)_{\hbt{QFT}}=\left(\B(z)_{{\hbt{QFT}}}^*\B(z)_{\hbt{QFT}}\right)^\met
	=\frac14\left[
	\begin{array}{cccc}
		r_1 & 0 & 0 & 0 \\
		0 & r_2 & 0 & 0 \\
		0 & 0 & r_2 & 0 \\
		0 & 0 & 0 & r_2 \\
	\end{array}\right]
	\vq\eeq\beq
	e^{i\BB(\theta)_{\hbt{QFT}}}
	=\left[
	\begin{array}{cccc}
		1 & 0 & 0 & 0 \\
		0 & 0 & 0 & -1 \\
		0 & 0 & -1 & 0 \\
		0 & -1 & 0 & 0 \\
	\end{array}\right]
	\eeq
	
	The (real) CM results in Eq. (\ref{example3}).
	\begin{figure*}[!t]
		\beq \label{example3}
		%\hskip-20mm
		\B(V)=
		\scalebox{0.85}{$
			\left[
			\begin{array}{cccccccc}
				e^{2 r_1}+3 e^{-2 r_2} & e^{2 r_1}-e^{-2 r_2} & e^{2 r_1}-e^{-2 r_2} & e^{2 r_1}-e^{-2 r_2} & 0 & 0 & 0 & 0 \\
				e^{2 r_1}-e^{-2 r_2} & e^{2 r_1}+3 e^{-2 r_2} & e^{2 r_1}-e^{-2 r_2} & e^{2 r_1}-e^{-2 r_2} & 0 & 0 & 0 & 0 \\
				e^{2 r_1}-e^{-2 r_2} & e^{2 r_1}-e^{-2 r_2} & e^{2 r_1}+3 e^{-2 r_2} & e^{2 r_1}-e^{-2 r_2} & 0 & 0 & 0 & 0 \\
				e^{2 r_1}-e^{-2 r_2} & e^{2 r_1}-e^{-2 r_2} & e^{2 r_1}-e^{-2 r_2} & e^{2 r_1}+3 e^{-2 r_2} & 0 & 0 & 0 & 0 \\
				0 & 0 & 0 & 0 & e^{-2 r_1}+3 e^{2 r_2} & e^{-2 r_1}-e^{2 r_2} & e^{-2 r_1}-e^{2 r_2} & e^{-2 r_1}-e^{2 r_2} \\
				0 & 0 & 0 & 0 & e^{-2 r_1}-e^{2 r_2} & e^{-2 r_1}+3 e^{2 r_2} & e^{-2 r_1}-e^{2 r_2} & e^{-2 r_1}-e^{2 r_2} \\
				0 & 0 & 0 & 0 & e^{-2 r_1}-e^{2 r_2} & e^{-2 r_1}-e^{2 r_2} & e^{-2 r_1}+3 e^{2 r_2} & e^{-2 r_1}-e^{2 r_2} \\
				0 & 0 & 0 & 0 & e^{-2 r_1}-e^{2 r_2} & e^{-2 r_1}-e^{2 r_2} & e^{-2 r_1}-e^{2 r_2} & e^{-2 r_1}+3 e^{2 r_2} \\
			\end{array}\right]
			$}
		\eeq
		\hrulefill
		\vspace*{0pt}
	\end{figure*}
	After the application of the cvQFT one finds (\ref{example4}).
	\begin{figure*}[!t]
		\beq \label{example4}
		\B(V)^{\hbt{QFT}}=4\left[
		\begin{array}{cccccccc}
			e^{2 r_1} & 0 & 0 & 0 & 0 & 0 & 0 & 0 \\
			0 &  \cosh (2 r_2) & 0 & - \sinh (2 r_2) & 0 & 0 & 0 & 0 \\
			0 & 0 &  e^{-2 r_2} & 0 & 0 & 0 & 0 & 0 \\
			0 & - \sinh (2 r_2) & 0 &  \cosh (2 r_2) & 0 & 0 & 0 & 0 \\
			0 & 0 & 0 & 0 &  e^{-2 r_1} & 0 & 0 & 0 \\
			0 & 0 & 0 & 0 & 0 &  \cosh (2 r_2) & 0 &  \sinh (2 r_2) \\
			0 & 0 & 0 & 0 & 0 & 0 &  e^{2 r_2} & 0 \\
			0 & 0 & 0 & 0 & 0 &  \sinh (2 r_2) & 0 &  \cosh (2 r_2) \\
		\end{array}\right]
		\eeq
		\hrulefill
		\vspace*{0pt}
	\end{figure*}

	\section{CONCLUSIONS}
	Considering the importance of the dvQFT for its very many applications 
	in several
	fields, we have introduced the QFT for continuous variables. The 
	dvQFT is applied to
	qubits and therefore it seems to be natural to search for an 
	extension to continuous
	variables, where the qubits are replaced by Gaussian states. In this 
	search we have
	found that the appropriate definition must be given in terms of 
	rotation operators,
	whose unitary matrix is given by the DFT matrix. We have introduced the acronym
	cvQFT to indicate this new form of Fourier transform.
	Once given the general definition of cvQFT, we have established its 
	properties and
	especially we investigated its implementation with primitive components
	(single--mode rotations and beam splitters). This topic is well known 
	and deeply investigated
	in the literature under the topic of factorization of unitary complex 
	matrices. The
	Murnaghan algorithm seems to be the best solution for the implementation.
	As done successfully for the cvQFT, we have investigated the fast implementation
	of the cvQFT. Using the techniques of the digital signal processing, 
	we have formulated a very efficient implementation for the cvQFT. 
	
	In the second part, we analyze how cvQFT acts on Gaussian operations 
	and states, showing that adding a cvQFT after a 
	displacement–rotation–squeezing cascade simply Fourier–transforms 
	the displacement vector and squeeze matrix and applies a Fourier–like 
	similarity transform to the rotation matrix. These results suggest 
	that cvQFT may serve as a natural building block in the design and 
	analysis of multimode Gaussian networks, entanglement–generation 
	schemes, and continuous–variable signal–processing protocols where 
	Fourier–type mode mixing is required.

	\bibliographystyle{IEEEtran}
	\bibliography{qfourier}
	
\end{document}